\documentclass[a4paper,UKenglish,cleveref,numberwithinsect, autoref, thm-restate]{lipics-v2021}

\usepackage{mathtools, mathpartir}
\usepackage{tikz}
\usepackage{tikz-cd}
\usepackage{scalerel}
\usepackage{dsfont}

\hypersetup{
  colorlinks=true,
  linkcolor=black!60!blue,
  citecolor=purple!60!black,
  urlcolor=cyan!60!blue
}

\newcommand{\B}{\mathsf{B}}
\newcommand{\I}{\mathsf{I}}
\newcommand{\C}{\mathsf{C}}
\newcommand{\W}{\mathsf{W}}
\newcommand{\K}{\mathsf{K}}
\newcommand{\D}{\mathsf{D}}
\newcommand{\F}{\mathsf{F}}
\newcommand{\Scom}{\mathsf{S}}

\newcommand{\Base}{\mathcal{B}}
\newcommand{\Total}{\mathcal{E}}
\newcommand{\fib}{\mathsf{p}}
\newcommand{\comp}{\chi}
\newcommand{\Arrow}[1]{#1^{\rightarrow}}
\newcommand{\cod}{\mathsf{cod}}
\newcommand{\proj}[1]{\pi_{#1}}
\newcommand{\CtxExt}[2]{#1 . #2}
\newcommand{\subst}[1]{#1^*}
\newcommand{\fiber}[2]{#1[#2]}

\newcommand{\disp}[1]{\overline{#1}}
\newcommand{\xx}{\disp{x}}
\newcommand{\yy}{\disp{y}}
\newcommand{\ff}{\disp{f}}

\newcommand{\TotalL}{\mathcal{L}}
\newcommand{\fibM}{\mathsf{q}}
\newcommand{\linL}{\mathsf{L}}
\newcommand{\linM}{\mathsf{M}}

\newcommand{\LCA}{A}
\newcommand{\LCAtoPCA}[1]{#1_!}

\newcommand{\Asm}[1]{\mathsf{Asm}(#1)}
\newcommand{\DAsm}[1]{\mathsf{DAsm}_{\mathsf{C}}(#1)}
\newcommand{\DAsmPr}{p}
\newcommand{\LAsm}[1]{\mathsf{DAsm}_{\mathsf{L}}(#1)}
\newcommand{\LAsmPr}{q}
\newcommand{\total}{\int}
\newcommand{\AsmL}{\mathsf{L}}
\newcommand{\AsmM}{\mathsf{M}}
\newcommand{\realize}{\Vdash}

\newcommand{\lin}{\multimap}
\newcommand{\tensor}{\otimes}

\newcommand{\Sqcap}{%
  \mathop{%
    \mathchoice
      {\scalebox{2.0}{$\sqcap$}}
      {\scalebox{1.6}{$\sqcap$}}
      {\scalebox{1.2}{$\sqcap$}}
      {\scalebox{1.0}{$\sqcap$}}
  }\displaylimits
}

\newcommand{\Sqsubset}{%
  \mathop{%
    \mathchoice
      {\scalebox{2.0}{$\sqsubset$}}
      {\scalebox{1.6}{$\sqsubset$}}
      {\scalebox{1.2}{$\sqsubset$}}
      {\scalebox{1.0}{$\sqsubset$}}
  }\displaylimits
}

\newcommand{\LSigSym}{\mathop{\Sqsubset}}
\newcommand{\LSig}[2]{\LSigSym_{#1} #2}
\newcommand{\LPiSym}{\mathop{\Sqcap}}
\newcommand{\LPi}[2]{\LPiSym_{#1} #2}
\newcommand{\LEq}[2]{\mathsf{Eq}(#1, #2)}
\newcommand{\Univ}{\mathcal{U}}
\newcommand{\El}{\mathsf{El}}
\newcommand{\ElC}[1]{\El_{\mathsf{C}}\left(#1\right)}
\newcommand{\ElL}[1]{\El_{\mathsf{L}}\left(#1\right)}
\newcommand{\expL}[1]{\mathbin{!} #1}
\newcommand{\FunExt}[1]{\mathsf{FunExt}(#1)}

\newcommand{\cartyp}[1]{{#1}\,\mathsf{type}}
\newcommand{\lintyp}[1]{{#1}\,\mathsf{lin}}

\newcommand{\projl}[1]{\pi_1 \> #1}

\newcommand{\ListI}[1]{\mathsf{List}^*(#1)}
\newcommand{\List}[1]{\mathsf{List}(#1)}
\newcommand{\NilI}{\mathsf{nil}^*}
\newcommand{\Nil}{\mathsf{nil}}
\newcommand{\ConsSym}{::}
\newcommand{\ConsI}[2]{#1 \ConsSym^* #2}
\newcommand{\Cons}[2]{#1 \ConsSym #2}
\newcommand{\Rec}[4]{\mathsf{rec}_{#1}(#2, #3, #4)}
\newcommand{\Alg}{\mathsf{Alg}}
\newcommand{\IsMor}[3]{\mathsf{IsMor}(#1, #2, #3)}
\newcommand{\Mor}[2]{\mathsf{Mor}(#1, #2)}

\newcommand{\defeq}{\mathrel{:\equiv}}
\newcommand{\judeq}{\mathrel{\equiv}}

\newcommand{\lam}[2]{\lambda{#1}.\,{#2}}

\newcommand{\lineq}[2]{#1 =_\mathsf{L} #2}

\newcommand{\m}{\mathsf{M}}
\newcommand{\lmod}{\mathsf{L}}

\pdfoutput=1 
\hideLIPIcs  

\title{Impredicativity in Linear Dependent Type Theory} 

\author{Sam L. Speight}{School of Computer Science, University of Birmingham, UK \and \url{https://slspeight.github.io/} }{s.l.speight@bham.ac.uk}{https://orcid.org/0009-0009-3706-9427}
{
This research was supported by the Engineering
and Physical Sciences Research Council grant EP/W034514/1
}

\author{Niels van der Weide}
{iCIS, Radboud University Nijmegen, The Netherlands \and \url{https://nmvdw.github.io}}
{nweide@cs.ru.nl}
{https://orcid.org/0000-0003-1146-4161}
{
This research was supported by the NWO project
``The Power of Equality'' OCENW.M20.380,
which is financed by the Dutch Research Council (NWO).
}

\authorrunning{S. Speight and N. van der Weide} 

\Copyright{S. Speight and N. van der Weide} 



\category{} 

\relatedversion{} 

\supplement{}
\supplementdetails[cite={linear-realizability-code}]{Software}{https://doi.org/10.5281/zenodo.18431109}

\acknowledgements{
The first author thanks Samson Abramsky for first introducing them to LCAs.
The second author thanks Dan Frumin for helpful discussions.
}

\nolinenumbers 

\EventEditors{John Q. Open and Joan R. Access}
\EventNoEds{2}
\EventLongTitle{42nd Conference on Very Important Topics (CVIT 2016)}
\EventShortTitle{CVIT 2016}
\EventAcronym{CVIT}
\EventYear{2016}
\EventDate{December 24--27, 2016}
\EventLocation{Little Whinging, United Kingdom}
\EventLogo{}
\SeriesVolume{42}
\ArticleNo{23}

\begin{document}

\maketitle

\begin{abstract}
    We construct a realizability model of linear dependent type theory from a linear combinatory algebra. Our model motivates a number of additions to the type theory. In particular, we add a universe with two decoding operations: one takes codes to cartesian types and the other takes codes to linear types. The universe is impredicative in the sense that it is closed under both large cartesian dependent products and large linear dependent products. We also add a rule for injectivity of the modality turning linear terms into cartesian terms. With all of the additions, we are able to encode (linear) inductive types. As a case study, we consider the type of lists over a linear type, and demonstrate that our encoding has the relevant uniqueness principle. The construction of the realizability model is fully formalized in the proof assistant Rocq.
\end{abstract}

\section{Introduction}\label{sec:intro}
Various proof assistants,
such as Rocq~\cite{therocqdevelopmentteam:2025}
and Lean~\cite{demoura:2015}, are based on \textbf{higher-order type theory}.
Higher-order type theories, such as the calculus of constructions~\cite{coquand:1988a},
are powerful systems that provide a suitable foundation for both mathematics and computer science, that is, for both proving and programming.
One of their key features is impredicativity.

Impredicativity comes in many forms,
and our focus on an \textbf{impredicative universe of types} in dependent type theory.
This means that the universe is closed under functions from any type to a small one.
Concisely, small types are closed under ``large'' products.


Impredicativity adds a lot strength to a logical system.
For instance,
such a system admits a translation from System F~\cite{Girard1972,Reynolds1974},
which allows one to use \textbf{impredicative encodings} (or polymorphic encodings) to reduce inductive types to (dependent) function types.
Such encodings guarantee that the resulting types satisfy the necessary recursion and $\beta$-principles.
However, they do not satisfy the relevant induction principle in general~\cite{geuvers:2001,rummelhoff:2004}.
To fix these defects,
various methods have been developed~\cite{altenkirch:2024,awodey:2018,marmaduke:2023,wadler:2003,wadler:2007}.

One challenge with higher-order type theories lies in finding semantics for them.
Since there is no (classical) set-theoretic model of higher-order type theories~\cite{reynolds:1984},
one needs alternative techniques to construct models for them.
\textbf{Realizability},
which dates back to Kleene~\cite{kleene:1945},
has established itself as one of the main methods to  do so~\cite{streicher:1991}.
Its applications range from program extraction~\cite{paulin-mohring:1989} to the construction of models that demonstrate the consistency of axioms
such as the formal Church's thesis~\cite{hyland:1982,oosten:2008}
and parametricity~\cite{krishnaswami:2013}.

Our goal in this paper is to extend the methods of realizability to higher-order type theories that incorporate \textbf{linearity} in their typing system.
Intuitively,
linearity says that functions are not allowed to delete or duplicate their arguments: 
they must use each of their arguments exactly once.
This discipline lends itself to various application domains,
for instance,
quantum programming languages~\cite{fu:2022}
and session types~\cite{DBLP:journals/pacmpl/ThiemannV20}.

Specifically,
we set out to to construct a realizability model of linear dependent type theory (LDTT)~\cite{cervesato1996linear,MR1946180,DBLP:conf/popl/KrishnaswamiPB15,vakar:2015,Vakar2017,lundfall:2017,Lundfall2018DiagramModel} that features impredicativity.
The model demostrates consistency of impredicativity with LDTT.
As an application of impredicativity in LDTT,
we use impredicative encodings
to construct an initial algebra of lists with elements taken from a linear type.

\subparagraph*{Formalization}
The results in \Cref{sec:linearrealizability,sec:linearrealizabilitymodel}
are formalized in the proof assistant Rocq~\cite{therocqdevelopmentteam:2025} using the UniMath library~\cite{unimath}.
We make heavy use of numerous results on category theory in UniMath~\cite{ahrens:2015,ahrens:2024a,wullaert:2022}.
Our formalization is publicly available online~\cite{linear-realizability-code},
and the file \texttt{Paper.v} indicates the relation between the paper and the formalized results.

Throughout this paper,
we use set-theoretic notation,
whereas our formalization is based on type theory.
While the largest part of the development remains unaffected by this difference in foundations,
there is one aspect worth mentioning:
in set theory, one defines Grothendieck fibration using functors;
whereas one uses displayed categories in type theory~\cite{ahrens:2019a}.

\subparagraph*{Contributions and outline}
The main contributions of this paper are as follows.
\begin{itemize}
    \item We construct a realizability model of linear dependent type theory with an impredicative universe of types (\Cref{sec:linearrealizabilitymodel}).
    \item We formalize this model in the proof assistant Rocq.
    \item We show how to use impredicativity in linear dependent type theory to construct the initial list algebra (\Cref{sec:encodings}).
\end{itemize}
In \Cref{sec:dltt} we discuss linear dependent type theory. The basic theory of (linear) combinatory algebras is recalled in \Cref{sec:linearrealizability}.
Related work is discussed in \Cref{sec:related}.

\section{Linear dependent type theory}\label{sec:dltt}
In this section, we recall the syntax of \textbf{linear dependent type theory} (LDTT),
and we introduce a number of additions that we use throughout this paper.
The version of LDTT that we use is closest to that of Lundfall \cite{lundfall:2017,Lundfall2018DiagramModel}, which in turn is based on the dual-context LDTT of Krishnaswami, Pradic and Benton \cite{DBLP:conf/popl/KrishnaswamiPB15}.
The different forms of judgement for linear dependent type theory are laid out in \autoref{tab:judgements}. There are also judgements for judgemental equality of contexts, types and terms, which we omit. Note that we have both cartesian contexts and mixed contexts; the latter contains a cartesian part and a linear part (though either may be empty). Types (whether cartesian or linear) are only well-formed in a cartesian context, hence cannot depend on linear variables. Cartesian terms are typed in a cartesian context; whereas linear terms are typed in a mixed context.

\begin{table}[t!]
    \centering
    \resizebox{\textwidth}{!}{%
    \begin{tabular}{||l l||} 
        \hline
        Judgement & Intended meaning \\ [0.5ex] 
        \hline\hline
        $\vdash \Gamma \, \mathsf{context}$
        &$\text{$\Gamma$ is a well-formed cartesian context}$
        \\
        $\vdash \Gamma \mid \Xi \, \mathsf{context}$
        &$\text{$\Gamma\mid\Xi$ is a well-formed mixed context, with cartesian part $\Gamma$ and linear part $\Xi$}$
        \\
        $\Gamma \vdash A \, \mathsf{type}$
        &$\text{$A$ is a well-formed cartesian type in cartesian context $\Gamma$}$
        \\
        $\Gamma \vdash \lintyp{A}$
        &$\text{$A$ is a well-formed linear type in cartesian context $\Gamma$}$
        \\
        $\Gamma \vdash a: A$
        &$\text{$a$ is a cartesian term of cartesian type $A$ in cartesian context $\Gamma$}$
        \\
        $\Gamma \mid \Xi \vdash a: A$
        &$\text{$a$ is a linear term of linear type $A$ in mixed context $\Gamma\mid\Xi$}$
        \\ [1ex] 
        \hline
    \end{tabular}
    }
    \caption{The different forms of judgement in linear dependent type theory}
    \label{tab:judgements}
\end{table}

Our type theory comes with various type formers.
On the cartesian side, we have unit-, $\prod$- and $\sum$-types, as well as identity types of cartesian types. We assert equality reflection for these identity types, collapsing propositional equality and definitional equality.
The rules for these constructors are the usual ones.

On the linear side, we have non-dependent multiplicative connectives (e.g. $\lin$, $\otimes$ and its unit), and additive connectives (e.g. $\oplus$, $\&$ and their units). These have the usual rules, familiar from linear type theory, formulated over an ambient cartesian context. Furthermore, we have linear analogues of $\prod$ and $\sum$, respectively, $\Sqcap$ and $\Sqsubset$. To give a flavour, the rules for $\Sqcap$ are provided in \autoref{fig:linpirules}. Note that the indexing type is required to be cartesian; the same is true for $\Sqsubset$.
\begin{figure*}[t!]
    \centering
    \resizebox{0.85\textwidth}{!}{%
    \makebox[\textwidth]{$
    \begin{aligned}
    \inferrule*[right=\text{$\Sqcap$}-f]
    {\Gamma \vdash \cartyp{A}
    \\
    \Gamma, x:A \vdash \lintyp{B}}
    {\Gamma \vdash \lintyp{\Sqcap_{x:A}B}}
    \qquad
    \inferrule*[right=\text{$\Sqcap$}-i]
    {\Gamma, x:A \mid \Xi \vdash b: B}
    {\Gamma \mid \Xi \vdash \lam{x}{b}:\Sqcap_{x:A}B}
    \qquad
    \inferrule*[right=\text{$\Sqcap$}-e]
    {\Gamma \mid \Xi \vdash t:\Sqcap_{x:A}B
    \\
    \Gamma \vdash a:A}
    {\Gamma \mid \Xi \vdash t \> a:B[a/x]}
    \end{aligned}
    $}
    }
    \\[0.5em]
    \resizebox{0.85\textwidth}{!}{%
    \makebox[\textwidth]{$
    \begin{aligned}
    \inferrule*[right=\text{$\Sqcap$}-\text{$\beta$}]
    {\Gamma, x:A \mid \Xi \vdash b: B
    \\
    \Gamma \vdash a: A}
    {\Gamma \mid \Xi \vdash (\lam{x}{b}) \> a \equiv b[a/x]: B[a/x]}
    \qquad
    \inferrule*[right=\text{$\Sqcap$}-\text{$\eta$}]
    {\Gamma \mid \Xi \vdash b: \Sqcap_{x:A}B}
    {\Gamma \mid \Xi \vdash \lam{x}{b \> x} \equiv b: \Sqcap_{x:A}B}
    \end{aligned}
    $}
    }
\caption{Rules for linear dependent functions}
\label{fig:linpirules}
\end{figure*}

Connecting the certesian and linear sides, we have two \textbf{modalities}. The modality $\m$ turns linear types and terms into cartesian types and terms. We can only apply $\m$ to a linear term in the empty linear context. The modality $\lmod$ turns linear types and terms into cartesian types and terms. We provide the rules for $\m$ in \autoref{fig:mrules}.
\begin{figure*}[t]
    \centering
    \resizebox{0.85\textwidth}{!}{%
    \makebox[\textwidth]{$
    \begin{aligned}
    \inferrule*[right=\text{$\m$}-f]
    {\Gamma \vdash \lintyp{A}}
    {\Gamma \vdash \cartyp{\m(A)} }
    \qquad
    \inferrule*[right=\text{$\m$}-i]
    {\Gamma \mid \cdot \vdash a: A}
    {\Gamma \vdash \m(a):\m(A)}
    \qquad
    \inferrule*[right=\text{$\m$}-e]
    {\Gamma \vdash a: \m(A)}
    {\Gamma \mid \cdot \vdash \mu(a): A}
    \end{aligned}
    $}
    }
    \\[0.5em]
    \resizebox{0.85\textwidth}{!}{%
    \makebox[\textwidth]{$
    \begin{aligned}
    \inferrule*[right=\text{$\m$}-\text{$\beta$}]
    {\Gamma \mid \cdot \vdash a:A}
    {\Gamma \mid \cdot \vdash \mu(\m(a)) \equiv a: A}
    \qquad
    \inferrule*[right=\text{$\m$}-\text{$\eta$}]
    {\Gamma \vdash a:\m(A)}
    {\Gamma \vdash \m(\mu(a)) \equiv a: \m(A)}
    \end{aligned}
    $}
    }
\caption{Rules for the modality $\m$ sending linear types and terms to cartesian types and terms}
\label{fig:mrules}
\end{figure*}


The remaining rules that we postulate are justified by our realizability model in \autoref{sec:linearrealizabilitymodel}, and, to our knowledge, have not previously been considered in the setting of linear dependent type theory. All of them play a part in the impredicative encodings of \autoref{sec:encodings}.

\subparagraph*{Equalizers}\label{sec:equalizers}
The first set of rules corresponds to having equalizers on the linear side. The formation rule is:
\[
\inferrule*[right=eq-f]
{\Gamma \mid \cdot \vdash f: A \lin B
\\
\Gamma \mid \cdot \vdash g: A \lin B
}
{\Gamma \vdash \lintyp{\LEq{f}{g}} }
\]
Next we have the introduction rule, which corresponds to the existence part of the universal property of equalizers.
\[
\inferrule*[right=eq-i]
{\Gamma \mid \Xi \vdash h: C \lin A
\\
\Gamma \mid \Xi \vdash f \circ h \equiv g \circ h: C \lin B
}
{\Gamma \mid \Xi \vdash \overline{h}: C \lin \LEq{f}{g}}
\]
The elimination rule supplies the equalizer map.
\[
\inferrule*[right=eq-e]
{ }
{\Gamma \mid \Xi \vdash \iota_{f,g}: \LEq{f}{g} \lin A}
\]
We have two $\beta$-rules.
The first $\beta$-rule expresses that the maps $f$ and $g$ are equalized by the map $\iota_{f,g}$. Note that they are equalized up to judgemental equality.
The second $\beta$-rule corresponds to the commutativity condition on the universal morphism.
\begin{align}
    &\Gamma \mid \Xi \vdash f \circ \iota_{f,g} \judeq g \circ \iota_{f,g}: \LEq{f}{g} \lin B
    \tag{\textsc{eq-\text{$\beta$\textsubscript{1}}}}
    \\
    &\Gamma \mid \Xi \vdash \iota_{f,g} \circ \overline{h} \judeq h: C \lin A
    \tag{\textsc{eq-\text{$\beta$\textsubscript{2}}}}
\end{align}
Finally, we have an $\eta$-rule that corresponds to uniqueness of the universal morphism.
\begin{align*}
    \inferrule*[right=eq-\text{$\eta$}]
    {\Gamma \mid \Xi \vdash h, h': C \lin \LEq{f}{g}
    \\
    \Gamma \mid \Xi \vdash \iota_{f,g} \circ h \judeq \iota_{f,g} \circ h' : C \lin A
    }
    {\Gamma \mid \Xi \vdash h \judeq h' : C \lin \LEq{f}{g}}
\end{align*}


\subparagraph*{Injectivity}

We add a rule for injectivity of the $\m$ modality.
\[
\inferrule*[right=\text{$\m$}-inj]
{\Gamma \mid \cdot \vdash \m(a) \judeq \m(a'): \m(A)
}
{\Gamma \mid \cdot \vdash a \judeq a': A}
\]
Semantically, this corresponds to faithfulness of the functor $\m$ (see \autoref{exa:lca-linear-comp-cat}). Note that the linear context is empty, as we are only allowed to apply $\m$ in an empty linear context.

Identity types of cartesian types together with the modality $\m$ facilitates derived identity types of \textit{linear} types (c.f. \cite[Section 6.2]{lundfall:2017}). We write $\lineq{a}{a'}$ to mean $\m(a) = \m(a')$, where we use $=$ for identity types of cartesian types. Thus we may only consider identity types of two linear terms in the empty linear context.

Combining $\m$-\textsc{inj} with equality reflection yields equality reflection for derived linear identity types: Suppose that we have a proof $\lineq{a}{a'}$. This is, by definition, a proof that $\m(a) = \m(a')$. By equality reflection we can deduce $\m(a) \equiv \m(a')$. Finally, using $\m$-\textsc{inj}, we conclude that $\m(a)\equiv \m(a')$. We will make use of linear equality reflection in \autoref{sec:encodings}.

We derive the following principle of function extensionality for $\Sqcap$. To obtain this, we first derive a judgemental version using \text{$\Sqcap$}-\text{$\eta$}, then we use equality reflection to get the propositional version:
\[
\inferrule*[right=\text{$\Sqcap$}-FunExt]
{\Gamma \mid \cdot \vdash f, g : \Sqcap_{x : A} B
\\
\Gamma, x : A \vdash p : \lineq{f \> x}{g \> x}}
{\Gamma \vdash \FunExt{p} : \lineq{f}{g} }
\]
Similarly,
we have function extensionality for non-dependent linear function types.

\subparagraph*{Impredicative universe}
Finally, we have rules are for an impredicative Tarski-style universe $\Univ$.
Note that $\Univ$ itself is a cartesian type, as we want to be able to use variables of type $\Univ$ (i.e. codes for types) in an unrestricted manner.
\[
\inferrule*[right=\text{$\Univ$}-f]
{\vdash \Gamma \, \text{context}
}
{\Gamma \vdash \cartyp{\Univ}}
\]
The universe comes equipped with two decoding operations: the first takes codes to cartesian types, and the second takes codes to linear types.
This is in contrast to having separate cartesian and linear universes.
\[
\inferrule*[right=\text{$\mathsf{El_C}$}]
{\Gamma \vdash A: \Univ
}
{\Gamma \vdash \cartyp{\ElC{A}}}
\qquad
\inferrule*[right=\text{$\mathsf{El_L}$}]
{\Gamma \vdash A: \Univ
}
{\Gamma \vdash \lintyp{\ElL{A}}}
\]

The universe is impredicative in the sense that it is closed under ``large'' cartesian dependent products \textit{and} ``large'' linear dependent products.
\[
\inferrule*[right=\text{$\widehat{\prod}$}-f]
{\Gamma \vdash \cartyp{A} \\
\Gamma, x:A \vdash B:\Univ
}
{\Gamma \vdash \mathop{\widehat{\prod}}_{x:A} B : \Univ}
\qquad
\inferrule*[right=\text{$\widehat{\Sqcap}$}-f]
{\Gamma \vdash \cartyp{A} \\
\Gamma, x:A \vdash B:\Univ
}
{\Gamma \vdash \mathop{\widehat{\Sqcap}}_{x:A} B : \Univ}
\]
To distinguish these ``$\Univ$-internal'' type formers from the ``external'' type formers, we have used the $\widehat{(-)}$ decoration. We stipulate that associated internal and external type formers are related by an isomorphism:
\begin{align}\label{eqn:el-isos}
    \ElC{\mathop{\widehat{\prod}}_{x:A} B} \cong \prod_{x:A} \ElC{B}
    \qquad
    \ElL{\mathop{\widehat{\Sqcap}}_{x:A} B} \cong \Sqcap_{x:A} \ElL{B}
\end{align}

The above isomorphisms can be used to derive abstraction and application operators. When there is no ambiguity, such as when we are working entirely within $\Univ$, we will use the usual notation (without the decoration) for these operators.
For the impredicative encodings in \autoref{sec:encodings}, we also require that the universe is closed under linear functions and equalizers.


\section{Linear realizability}\label{sec:linearrealizability}
In this section,
we recall the notion of \textbf{linear combinatory algebra}~\cite{AbramskyHaghverdiScott2002LCAs}.
Before we do so,
we first look at the notions of \textbf{applicative structure} and \textbf{combinatory algebra}~\cite{Schonfinkel1924}.

\begin{definition}
    An \textbf{applicative structure} is a set $A$ together with a binary operation $(-)\bullet(=):A\times A \rightarrow A$,
    which we call \textbf{application}.
\end{definition}
We often omit the $\bullet$, and we always associate to the left, so that, e.g. $abc = (a\bullet b)\bullet c$.

The objects defined in the following are usually called \textbf{combinatory algebras}. Here we prepend the adjective cartesian to emphasise the difference between these and \textit{linear} combinatory algebras, to be introduced shortly.
\begin{definition}[Schönfinkel \cite{Schonfinkel1924}]
    A \textbf{cartesian combinatory algebra} (CCA) is an applicative structure $(A,\bullet)$ in which there exist combinators $\Scom',\K'\in A$
    such that $\Scom' a b c = ac(bc)$ and $\K' a b = a$.
\end{definition}
We decorate combinators in CCAs with primes in order to distinguish them from combinators in \textit{linear} combinatory algebras.

The following definition of \textit{linear} combinatory algebra \cite{AbramskyHaghverdiScott2002LCAs} corresponds to a Hilbert-style axiomatisation of the $!,\lin$-fragment of linear logic.
\begin{definition}
    A \textbf{linear combinatory algebra} (LCA) is an applicative structure $(A,\bullet)$, together with a unary operation $!:A\rightarrow A$, in which there exist ``combinators''
    $\B,\I,\C,\W,\K,\D,\delta,\F\in A$
    satisfying the identities in \Cref{tab:combinators}.
\end{definition}

Examples of linear combinator algebras can be found in the literature~\cite{AbramskyHaghverdiScott2002LCAs,hoshino:2007}.
Another convention we follow is that $!$ binds more tightly than $\bullet$. Observe that certain combinators compute only when particular arguments are guarded by $!$.
\begin{table}[t!]
    \centering
    \begin{tabular}{||l l l||} 
        \hline
        Combinator identity & Principal type & Logical principle \\ [0.5ex] 
        \hline\hline
        $\B abc=a(bc)$ & $(\alpha\lin\beta)\lin(\gamma\lin\alpha)\lin\gamma\lin\beta$ & cut/composition \\ 
        $\I a = a$ & $\alpha\lin\alpha$ & identity \\
        $\C abc = acb$ & $(\alpha\lin\beta\lin\gamma)\lin\beta\lin\alpha\lin\gamma$ & exchange \\
        $\W a!b = a!b!b$ & $(\expL{\alpha}\lin \expL{\alpha}\lin\beta)\lin \expL{\alpha}\lin\beta$ & contraction \\
        $\K a!b = a$ & $\alpha\lin \expL{\beta} \lin\alpha$ & weakening \\
        $\D!a = a$ & $\expL{\alpha} \lin\alpha$ & dereliction \\
        $\delta!a = !!a$ & $\expL{\alpha} \lin \expL{!\alpha}$ & comultiplication \\
        $\F!a!b = !(ab)$ & $\expL{(\alpha\lin\beta)}\lin \expL{\alpha}\lin\expL{\beta}$ & functoriality \\ [1ex] 
        \hline
    \end{tabular}
    \caption{Combinator identities with their corresponding principal type and logical principle}
    \label{tab:combinators}
\end{table}

For an applicative structure, the existence of $\Scom'$ and $\K'$ combinators is equivalent to the existence of $\B'$, $\C'$, $\K'$ and $\W'$ combinators, where $\B'$, $\C'$ and $\W'$ satisfy the identities from \autoref{tab:combinators} with any occurrences of $!$ omitted. Note that we can define $\I'\coloneqq \W'\K'$.

As shown by Abramsky, Haghverdi and Scott \cite{AbramskyHaghverdiScott2002LCAs,Haghverdi2000}, from any LCA $A = (A,\bullet,!)$, we may construct a CCA $\LCAtoPCA{A} \coloneq (A,\bullet_!)$ on the same underlying set where $a\bullet_! b \coloneqq a \mathrel{\bullet} {!b}$.

\subparagraph*{Combinatory completeness}

The key property of CCAs is their ability to mimic $\lambda$-abstraction. This property is known as \textbf{combinatory completeness}. LCAs enjoy their own form of combinatory completeness, which relates to the linear $\lambda$-calculus. These properties are extremely useful when exhibiting realizers for our model in \autoref{sec:linearrealizabilitymodel}. 

Let $A$ be a set and $V$ be a countably infinite set of variables. Polynomials over $A$ are given by the first three clauses of the following grammar; $!$-polynomials are given by the whole grammar.
\begin{align*}
    t, t' ::= x\in V \,\mid\, a\in A \,\mid\, t\bullet t' \,\mid\, !t
\end{align*}
We adopt the same conventions for ($!$-)polynomials as for applicative systems with $!$. 

Closed polynomials can be identified with elements of $A$. If $A$ comes with a $!:A \rightarrow A$, then closed $!$-polynomials can be identified with elements of $A$. Two ($!$-)polynomials $t,t'$ with variables in $\{x_1,...,x_n\}$ are considered equal when for all $a_1,...,a_n\in A$, $t[a_1,...,a_n/x_1,...,x_n] = t'[a_1,...,a_n/x_1,...,x_n]$.
A variable occurring in a $!$-polynomial $t$ is \textit{linear} iff it appears exactly once and not in scope of a $!$.
We write $\mathsf{vars}(t)$ for the set of variables occurring in a ($!$-)polynomial $t$, and $\mathsf{linvars}(t')$ for the set of linear variabels occurring in a $!$-polynomial $t'$.

\begin{proposition}[Curry \cite{Curry1930}]\label{thm:combinatorycompleteness-cartesian}
    The following are equivalent for an applicative system $(A,\bullet)$:
    \begin{itemize}
        \item $(A,\bullet)$ is a cartesian combinatory algebra.
        \item $(A,\bullet)$ is cartesian combinatory complete. That is, for every polynomial $t$ over $A$, there is a polynomial $\lambda^* x.\, t$ over $A$, with $\mathsf{vars}(\lambda^* x.\, t) = \mathsf{vars}(t)-\{x\}$, such that $(\lambda^* x.\, t)x = t$.
    \end{itemize}
\end{proposition}

\begin{proposition}[Simpson \cite{simpson2005reduction}]\label{thm:combinatorycompleteness-linear}
    Let $(A,\bullet)$ be an applicative system and $!:A \rightarrow A$. The following are equivalent:
    \begin{itemize}
        \item $(A,\bullet,!)$ is a linear combinatory algebra.
        \item $(A,\bullet,!)$ is linear combinatory complete. That is, for every $!$-polynomial $t$ over $A$:
        \begin{itemize}
            \item if $x\in\mathsf{linvars}(t)$, then there is a polynomial $\lambda^* x.\, t$ over $A$,
            with $\mathsf{vars}(\lambda^* x.\, t) = \mathsf{vars}(t)-\{x\}$
            and $\mathsf{linvars}(\lambda^* x.\, t) = \mathsf{linvars}(t)-\{x\}$,
            such that $(\lambda^* x.\, t)x = t$;
            \item there is a polynomial $\lambda!^* x.\, t$ over $A$,
            with $\mathsf{vars}(\lambda!^* x.\, t) = \mathsf{vars}(t)-\{x\}$
            and $\mathsf{linvars}(\lambda!^* x.\, t) = \mathsf{linvars}(t)-\{x\}$,
            such that $(\lambda!^* x.\, t)!x = t$.
        \end{itemize}
    \end{itemize}
\end{proposition}

\section{Linear realizability model}
\label{sec:linearrealizabilitymodel}
In this section, 
we construct the realizability model of linear dependent type theory. 
Traditional approaches to realizability start with a (cartesian) combinatory algebra,
and then a model of Martin-L\"of type theory is constructed from assemblies over that CCA~\cite{luo:1990,reus:1999a}.
Such a model is represented using some categorical structure,
and often, comprehension categories are used~\cite{jacobs:2001}.
We set out to do the same, but for LDTT using linear combinatory algebras.
In order to model LDTT,
we use a more refined categorical structure,
namely \textbf{linear comprehension categories},
which are based on work by Lundfall~\cite{lundfall:2017,Lundfall2018DiagramModel}.

We start this section by recalling comprehension categories (\Cref{def:comp-cat}) and the realizability model of Martin-L\"of type theory (\Cref{exa:assembly-comp-cat}).
After that,
we introduce linear comprehension categories (\Cref{def:linear-comp-cat}).
Throughout this section we show that every LCA gives rise to a linear comprehension categories, called the \textbf{linear realizability model}.
We discuss the relevant features step by step.
Specifically,
we first discuss the interpretation of cartesian types (\Cref{exa:assembly-comp-cat})
and linear types (\Cref{exa:monoidal-fibration-lca}),
and the modalities (\Cref{exa:lca-linear-comp-cat}).
We finish by interpreting various type formers (\Cref{exa:linear-realizability-types}),
including equalizer types (\Cref{exa:linear-realizability-equalizers}) and an impredicative universe (\Cref{exa:impred-universe}).




\subsection{Realizability models of type theory}
We begin with a brief overview of comprehension categories
and the usual realizability model of dependent type theory.
The key ingredient for comprehension categories is the notion of a \emph{Grothendieck fibration}~\cite{grothendieck:1960},
which is defined as follows.

\begin{definition}
\label{def:fibration}
Let $\fib : \Total \rightarrow \Base$ be a functor,
and let $f : x \rightarrow y$ be a morphism in $\Total$.
We say that $f$ is \textbf{cartesian}
if for all objects $z \in \Total$
and morphisms $h : z \rightarrow y$ and $g : \fib(z) \rightarrow \fib(x)$
such that $\fib(h) = \fib(f) \circ g$,
there is a unique $k : z \rightarrow y$ with $\fib(k) = g$
making the resulting triangle commute.
We say that $\fib$ is a \textbf{(Grothendieck) fibration}
if for all $f : x \rightarrow \fib(y)$ in $\Base$
there is an object $\xx \in \Total$ with $\fib(\xx) = x$
and a cartesian morphism $\ff : \xx \rightarrow y$
with $\fib(\ff) = f$.
\end{definition}

If we have a fibration $\fib : \Total \rightarrow \Base$,
then we can think of $\Base$ and $\Total$ as the categories of contexts and of types respectively.
Given an object $x \in \Base$,
we denote the fibre category of $\fib$ along $x$ by $\fiber{\fib}{x}$.
Every morphism $s : x \rightarrow y$ gives rise to a functor $\subst{s} : \fiber{\fib}{y} \rightarrow \fiber{\fib}{x}$,
which gives us a substitution operation on types.

Comprehension categories~\cite{jacobs:2001} combine fibrations with a comprehension functor,
which is used to interpret context extension.
We define them as follows.

\begin{definition}
\label{def:comp-cat}
A \textbf{comprehension category} is a strictly commuting diagram of functors as indicated below.
\[
\begin{tikzcd}
    \Total && {\Arrow{\Base}} \\
    & \Base
	\arrow["\comp", from=1-1, to=1-3]
	\arrow["\fib"', from=1-1, to=2-2]
	\arrow["\cod", from=1-3, to=2-2]
\end{tikzcd}
\]
We assume that $\fib$ is a fibration
and that $\comp$ preserves cartesian morphisms.
We denote such a comprehension category by $\comp : \Total \rightarrow \Arrow{\Base}$.
We say that $\comp : \Total \rightarrow \Arrow{\Base}$ is \textbf{full} if $\comp$ is fully faithful.
\end{definition}

If we have $\xx \in \Total$ with $\fib(\xx) = x$,
then we denote $\comp(\xx)$ by $\proj{\xx} : \CtxExt{x}{\xx} \rightarrow x$.
The fact that $\comp$ preserves cartesian morphisms
means that it maps each cartesian lift to a pullback square.

\begin{remark}
Throughout this work we only consider full comprehension categories.
If a comprehension category is not full,
then it comes with additional structure,
namely morphisms between types.
Such structure gives an interpretation of subtyping~\cite{coraglia:2024a,najmaei:2026},
which usually is not present in Martin-L\"of type theory.
Since we are only interested in models without subtyping,
we only consider full comprehension categories.

In addition,
we only look at comprehension categories for which $\fib$ has a fibrewise terminal object
and for which $\comp$ preserves terminal objects.
In such comprehension categories,
terms of type $\xx$ in context $x$ are the same as
morphisms in $\fiber{\fib}{x}$ from the terminal object to $x$.
\end{remark}

The interpretation of dependent products and sums are based on the idea by Lawvere
that quantifiers are adjoints~\cite{jacobs:2001,lawvere:2006}.

\begin{definition}
\label{def:quantifiers}
We say that a full comprehension category $\comp : \Total \rightarrow \Arrow{\Base}$ supports \textbf{$\sum$-types}
if for all for all $s : x \rightarrow y$ the functor $\subst{s}$ has a left adjoint
satisfying the usual Beck-Chevalley condition.
We denote this left adjoint by $\sum$.
If $\comp : \Total \rightarrow \Arrow{\Base}$ supports $\sum$-types,
then we say that these $\sum$-types are \textbf{strong}
if for all objects $x \in \Base$, $\xx \in \fiber{\fib}{x}$ and $\yy \in \fiber{\fib}{\CtxExt{x}{\xx}}$
the canonical morphism from $\CtxExt{\CtxExt{x}{\xx}}{\yy}$ to $\CtxExt{x}{\sum_{\xx} \yy}$ is an isomorphism.

A comprehension category $\comp : \Total \rightarrow \Arrow{\Base}$ supports \textbf{$\prod$-types}
if for all for all $s : x \rightarrow y$ the functor $\subst{s}$ has a right satisfying the usual Beck-Chevalley condition.
\end{definition}

\begin{remark}
There are some important observations to be made about \Cref{def:quantifiers}.
While one generally only considers adjoints to the weakening functors $\subst{\proj{\xx}}$,
we look at adjoints for all morphisms.
As a consequence, \Cref{def:quantifiers} captures not only $\sum$-types,
but also extensional identity types,
which are given by left adjoints to the contraction functors.
In addition,
one only needs to check the Beck-Chevalley condition for $\sum$-types,
because the one for $\prod$-types follows from it \cite{seely:1983}.
Finally,
we require our $\sum$-types to be strong in order to interpret the $\eta$-rule for dependent sums.
\end{remark}

Now we construct the realizability model of dependent type theory~\cite{luo:1990,reus:1999a}.

\begin{example}
\label{exa:assembly-comp-cat}
We show that each cartesian combinatory algebra $A$ gives rise to a model of extensional type theory with $\prod$ and $\sum$-types.
First, we define the category of contexts.
Recall that an assembly over $A$ is a pair of a set $X$ together with a relation ${\realize } \subseteq A \times X$ such that for each $x \in X$ there exists $a \in A$ with $a \realize x$.
We call $\realize$ the ``realizability relation'' for $X$.
Given assemblies $X$ and $Y$,
an assembly morphism is given by a function $f : X \rightarrow Y$
for which there exists $e\in A$ such that for each $a\in A$ with $a \realize x$
we have $e a \realize f(x)$.
We denote by $\Asm{A}$
the category whose objects are assemblies over $A$
and whose morphisms are assembly morphisms.
In fact, we can construct a category of assemblies for any applicative structure with a $\B$ and $\I$ combinator.

Next we define the category $\DAsm{A}$ of types.
Objects of this category are \textbf{families of assemblies},
which are pairs of an assembly $X$
together with an assembly $Y_x$ for each $x \in X$.
A morphism from $(X, Y)$ to $(X', Y')$
consists of an assembly morphism $f : X \rightarrow X'$
together with a map $g_x : Y_x \rightarrow Y_{f(x)}$ for each $x$
for which there exists $e \in A$
such that $e a b \realize g_x(y)$ whenever $a \realize x$ and $b \realize y$.
The first projection $\DAsmPr : \DAsm{A} \rightarrow \Asm{A}$ is a fibration.

Finally, we define a functor $\total : \DAsm{A} \rightarrow \Arrow{\Asm{A}}$.
If we have an assembly $X$ together with assemblies $Y_x$ for $x \in X$,
then we define $\total_{x : X} Y_x$ to be the assembly whose underlying set is $\sum_{x : X} Y_x$.
All in all, we get the following comprehension category.
\[
\begin{tikzcd}
	{\DAsm{A}} && {\Arrow{\Asm{A}}} \\
	& {\Asm{A}}
	\arrow["\total", from=1-1, to=1-3]
	\arrow["\DAsmPr"', from=1-1, to=2-2]
	\arrow["\cod", from=1-3, to=2-2]
\end{tikzcd}
\]
This comprehension supports $\sum$-types and $\prod$-types,
and the functor $\total$ is fully faithful and preserves terminal objects.
In addition,
each fibre category $\fiber{\DAsmPr}{x}$ has a terminal object.
\end{example}

\subsection{Linear comprehension categories}
Next we extend our notion of comprehension category to incorporate the features of linear type theory,
and we define linear comprehension categories.
There are two additional features to such comprehension categories.
The first feature is that
they come with two notion of types,
namely cartesian types and linear types.
We represent these two kinds via two fibrations over the same category.
To incorporate linearity,
we require one of these fibrations to be \textbf{symmetric monoidal closed}.

\begin{definition}
\label{def:monoidal-fibration}
Let $\fibM : \Total \rightarrow \Base$ be a fibration.
We say that $\fib$ is a \textbf{symmetric monoidal closed fibration}
if for each $x \in \Base$ the fibre $\fiber{\fibM}{x}$ is a symmetric monoidal closed category
and if for each $s : x \rightarrow y$ the functor $\subst{s} : \fiber{\fibM}{y} \rightarrow \fiber{\fibM}{x}$ is a symmetric monoidal closed functor.
\end{definition}

Note that there are various notions of monoidal fibrations in the literature \cite{moeller:2020,shulman:2008}
depending on whether the total category or the fibre categories have a monoidal structure.
For our purposes,
it is necessary to require a monoidal structure for the fibre categories.

From a type theoretic point of view,
we can understand symmetric monoidal closed fibrations as follows.
Each fibre category $\fiber{\fibM}{x}$ is seen as the category of types
in some context $x$.
The monoidal structures equip these categories with three type formers:
a unit, a monoidal product, and linear function types,
and thus each $\fiber{\fibM}{x}$ forms a model of linear logic.
Since the functors $\subst{s} : \fiber{\fibM}{y} \rightarrow \fiber{\fibM}{x}$ are symmetric monoidal closed,
the aforementioned type formers are stable under substitution up to isomorphism.

Using linear combinatory algebras,
we construct an example of a monoidal fibration. 

\begin{example}
\label{exa:monoidal-fibration-lca}
Let $\LCA$ be an LCA.
We define the category $\LAsm{\LCA}$ of families of linear assemblies over $\LCA$
and morphisms between them.
A \textbf{family of linear assemblies}
over $\LCA$ is a pair consisting of an assembly $X \in \Asm{\LCAtoPCA{\LCA}}$ over the CCA $\LCAtoPCA{A}$ associated to the LCA $\LCA$, together with an
assembly $Y_x : \Asm{\LCA}$ over the LCA $\LCA$ for each $x \in X$.
A morphism from $(X, Y)$ to $(X', Y')$
consists of a morphism $f : X \rightarrow X'$ in $\Asm{\LCAtoPCA{A}}$,
together with a map $g_x : Y_x \rightarrow Y_{f(x)}$ for each $x$, for which there exists $e \in A$
such that $e {\expL{a}} b \realize g_x(y)$ whenever $a \realize x$ and $b \realize y$.
The functor $\LAsmPr : \LAsm{\LCA} \rightarrow \Asm{\LCAtoPCA{\LCA}}$ is the first projection.

The functor $\LAsmPr$ is a symmetric monoidal closed fibration.
The unit for the monoidal structure on $\fiber{\LAsmPr}{X}$ is given by the unit set with $\I$ as the only realizer.
The monoidal product of linear assemblies $Y$ and $Y'$
has $Y(x) \times Y'(x)$ for $x \in X$ as carrier,
and we say that $a \realize (y, y')$ if there are $b, b' \in \LCA$
such that $b \realize y$, $b' \realize y'$,
and $a = \lambda^* x .\, x b b'$.
\end{example}


The second key feature of linear comprehension categories lies with the interpretation of linear exponentiation.
There are various ways to represent this connective categorically,
and our approach is based on linear non-linear models \cite{Benton1994LNL}.
Specifically,
we require that every linear comprehension category comes with an adjunction between the category of linear types and of cartesian types.

\begin{definition}
\label{def:linear-comp-cat}
A \textbf{linear comprehension category} is given by a strictly commuting diagram of functors as follows.
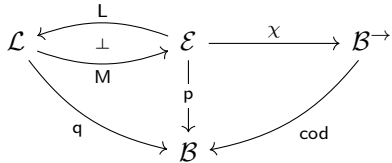
\[
\begin{tikzcd}[column sep = huge, row sep = large]
	\TotalL & \Total & {\Arrow{\Base}} \\
	& \Base
	\arrow[""{name=0, anchor=center, inner sep=0}, "\linM"', bend right=20, from=1-1, to=1-2]
	\arrow["\fibM"', bend right=20, from=1-1, to=2-2]
	\arrow[""{name=1, anchor=center, inner sep=0}, "\linL"', bend right=20, from=1-2, to=1-1]
	\arrow["\comp", from=1-2, to=1-3]
	\arrow["\fib"{description}, from=1-2, to=2-2]
	\arrow["\cod", bend left=20, from=1-3, to=2-2]
	\arrow["\dashv"{anchor=center, rotate=-90}, draw=none, from=0, to=1]
\end{tikzcd}
\]
We assume that $\fib$ is a fibration,
and that $\Total$ has fibrewise terminal objects and binary products.
We require $\comp$ to be fully faithful
and to preserve cartesian morphisms and terminal objects.
We assume that $\fibM$ is a symmetric monoidal closed fibration
and that $\linL \dashv \linM$ is a fibrewise symmetric monoidal adjunction.
Finally,
$\linL$ and $\linM$ need to preserve cartesian morphisms.
\end{definition}

The interpretation of linear dependent type theory in linear comprehension categories is indicated in \Cref{tab:interpretation}.
The most important observation to make is the interpretation of linear terms,
which are interpreted as morphisms in $\TotalL$,
whereas terms of cartesian types are interpreted as sections of projections.

\begin{table}[t!]
\centering
\begin{tabular}{||l l l||}
\hline
Syntax         & Linear Comprehension Category & Realizability Model                          \\
[0.5ex] 
\hline\hline
Context        & Object in $\Base$             & Assembly over  $\LCAtoPCA{\LCA}$            \\
Cartesian type & Object in $\Total$            & Family of assemblies over $\LCAtoPCA{\LCA}$ \\
Term of type $X$           & Section of $\proj{X}$            & Section of $\proj{X}$                           \\
Linear type    & Object in $\TotalL$           & Family of assemblies over $\LCA$             \\
Linear term    & Morphism in $\TotalL$         & Morphism of linear assemblies \\ [1ex] \hline
\end{tabular}
\caption{The interpretation of linear dependent type theory in linear comprehension categories, and in our realizability model in particular}
\label{tab:interpretation}
\end{table}

The functors $\linL$ and $\linM$ interpret the modalities in linear logic,
and their composition gives linear exponentiation.
Here it is important to note that $\linM$ also acts on linear terms,
but only those that use no linear variables.
This restriction also makes sense from a semantical point of view.
Suppose that we have some type $\xx \in \fiber{\fib}{x}$.
Recall that terms of type $\xx$ can equivalently be represented
as a morphism from the terminal object to $\xx$.
Since $\linM$ is strong monoidal,
linear terms in the empty linear context correspond to terms of type $\linM(\xx)$.

To construct examples of linear comprehension categories,
it is helpful to use the following characterization of symmetric monoidal adjunctions.

\begin{proposition}[{\cite[Lemma 13]{mellies:2003}}]
An adjunction $L \dashv R$ between symmetric monoidal categories $\mathcal{C}$ and $\mathcal{D}$ is symmetric monoidal if and only if $L$ is a strong monoidal functor.
\end{proposition}

\begin{example}
\label{exa:lca-linear-comp-cat}
The functor $\linM : \LAsm{\LCA} \rightarrow \DAsm{\LCAtoPCA{\LCA}}$ sends every a linear assembly $(X, Y)$ to the family $Y_x$ of assemblies for each $x \in X$.
Note that this definition makes sense,
because $\LCA$ and $\LCAtoPCA{\LCA}$ have the same underlying set of elements.
Next we define the functor $\linL : \DAsm{\LCAtoPCA{\LCA}} \rightarrow \LAsm{\LCA}$.
Suppose that we have a family of assemblies $Y_x$ for each $x \in X$.
For each $x \in X$, we define an assembly $Y_x'$ over $\LCA$,
whose underlying set is $Y_x$
and where $a \realize y$ if there is $b \in A$
such that $b \realize y$ and $a = \expL{b}$.
The functor $\linL$ maps the family $Y_x$
to the linear assembly $Y_x'$.
We arrive at the following linear comprehension category, where the cartesian side is \autoref{exa:assembly-comp-cat} instantiated with the CCA $\LCAtoPCA{A}$.
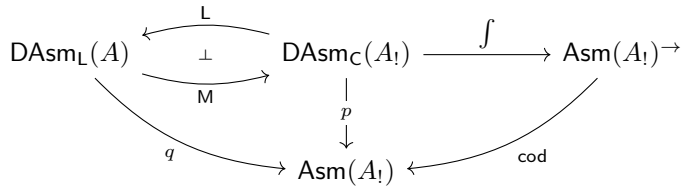
\[
\begin{tikzcd}[column sep = huge, row sep = large]
	{\LAsm{\LCA}} & {\DAsm{\LCAtoPCA{\LCA}}} & {\Arrow{\Asm{\LCAtoPCA{\LCA}}}} \\
	& {\Asm{\LCAtoPCA{\LCA}}}
	\arrow[""{name=0, anchor=center, inner sep=0}, "\AsmM"', bend right=15, from=1-1, to=1-2]
	\arrow[""{name=1, anchor=center, inner sep=0}, "\AsmL"', bend right=15, from=1-2, to=1-1]
    \arrow["\LAsmPr"', bend right=20, from=1-1, to=2-2]
	\arrow["\total", from=1-2, to=1-3]
	\arrow["\DAsmPr"{description}, from=1-2, to=2-2]
	\arrow["\cod", bend left=20, from=1-3, to=2-2]
    \arrow["\dashv"{anchor=center, rotate=-90}, draw=none, from=0, to=1]
\end{tikzcd}
\]
Both $\AsmM$ and $\AsmL$ are faithful.
This examples generalizes Hoshino's~\cite{hoshino:2007} linear-non-linear model to dependent types.
\end{example}

Before we look at type formers in the linear realizability model,
we first explicitly describe (linear) types and terms in that model.
Objects and morphisms in $\Asm{\LCAtoPCA{\LCA}}$ 
represent contexts and substitutions respectively.
Note that we use the combinatory algebra $\LCAtoPCA{\LCA}$,
because linear types depend intuitionistically on their variables.
Objects in $\DAsm{\LCAtoPCA{\LCA}}$ give us cartesian types,
and objects in $\LAsm{\LCA}$ give us linear types.
While both $\LAsm{\LCA}$ and $\DAsm{\LCAtoPCA{\LCA}}$ are fibrewise monoidal categories,
only $\DAsm{\LCAtoPCA{\LCA}}$ is guaranteed to be cartesian monoidal.
For this reason, $\LAsm{\LCA}$ only supports linear logic in general.

\subsection{Type formers in linear comprehension categories}
Finally,
we look at type formers in linear comprehension categories,
and in the linear realizability model in particular.
We start with the linear quantifiers,
which we interpret as adjoints of substitution.

\begin{definition}
We say that a linear comprehension category supports \textbf{linear dependent products} if for each $s : x \rightarrow y$ in $\Base$
the functor $\subst{s} : \fiber{\fibM}{y} \rightarrow \fiber{\fibM}{x}$
has a right adjoint satisfying the usual Beck-Chevalley condition.

A linear comprehension category supports \textbf{linear dependent sums}  if for each $s : x \rightarrow y$ in base
the functor $\subst{s} : \fiber{\fibM}{y} \rightarrow \fiber{\fibM}{x}$
has a left adjoint satisfying the usual Beck-Chevalley condition.
The left adjoint of $\subst{s}$ is denoted by $\LSigSym_{s}$.
We also require Frobenius reciprocity:
if we have a morphism $s : x \rightarrow y$
and objects $\xx : \fiber{\fibM}{x}$ and $\yy : \fiber{\fibM}{y}$,
then the canonical morphism
from $\LSig{s}{(\subst{s}(y) \tensor \xx)}$
to $\yy \tensor \LSig{s}{\xx}$
is an isomorphism.
\end{definition}

\begin{example}
\label{exa:linear-realizability-types}
The linear realizability model defined in \Cref{exa:lca-linear-comp-cat} supports linear dependent products and sums.
We only discuss the construction of linear dependent products.
Suppose that we have a morphism $s : X \rightarrow X'$ in $\Asm{\LCAtoPCA{\LCA}}$
and that for each $x \in X$
we have an assembly $Y_x$ over $\LCA$.
Let $x' \in X'$
and let $f$ be function that assigns to each $x \in s^{-1}(x')$
an element $f(x) \in Y_x$.
We say that $f$ is tracked by $a$
if $a {\expL{b}} \realize f(x)$
whenever $b \realize x$.
The carrier of $\LPi{s}{Y}$
is defined to be the set of all such $f$ that is tracked by some $a$,
and we say that $a \realize f$ if $a$ tracks $f$.
\end{example}

Next we look at equalizer types.
We interpret these using fibrewise limits.

\begin{definition}
Let $\fib : \Total \rightarrow \Base$ be a fibration.
We say that $\fib$ has a \textbf{fibrewise equalizers}
if each fibre $\fiber{\fib}{x}$ has equalizers
and if the functors $\subst{s} : \fiber{\fib}{y} \rightarrow \fiber{\fib}{x}$ preserve equalizers.
We say that a linear comprehension category supports \textbf{equalizer types}
if the fibration $\fibM : \TotalL \rightarrow \Base$ has fibrewise equalizers.
\end{definition}

If $\subst{s}$ has a left and right adjoint, it preserves limits and colimits, hence it suffices to check that each $\fiber{\fib}{x}$ has equalizers.
\begin{example}
\label{exa:linear-realizability-equalizers}
The linear comprehension category defined in \Cref{exa:lca-linear-comp-cat} supports equalizer types.
Equalizers in the fibre $\fiber{\LAsmPr}{x}$ are constructed in the usual way.
\end{example}

The additive connectives of linear logic are treated in the same way.
For instance, the additive conjunction is interpreted via fibrewise binary products.
Since we do not use these connectives in the remainder,
we do not discuss them in detail.

We finish our development by showing that our model supports an impredicative universe of types.
In realizability,
such universes build forth upon the notion of \emph{modest set},
and the key theorem is that modest sets are equivalent to \emph{partial equivalence relations}.
Note that this theorem holds for BI-algebras,
which are applicative structures that come with a $\B$ combinator and an $\I$ combinator.
This generality is important in our development,
because it allows us to apply it to both linear and cartesian types in our realizability model.

\begin{definition}
\label{def:modest-set}
Let $X$ be an assembly for a BI-algebra $A$.
We say that $X$ is \textbf{modest}
if we have $x = y$ whenever $a \realize x$ and $a \realize y$ for some $a$.
A \textbf{partial equivalence relation} (PER) is a relation on $A$ that is symmetric and transitive.
\end{definition}

\begin{theorem}[Hyland \cite{hyland:1988}]
\label{thm:modest-set-equiv}
Let $A$ be a BI-algebra.
The categories of modest sets over $A$ and of partial equivalence relations over $A$ are equivalent.
\end{theorem}

Note that \Cref{thm:modest-set-equiv} has been formalized before for partial combinatory algebras~\cite{chhabra:2024}.
The importance of \Cref{thm:modest-set-equiv} comes from that fact that 
we can represent the large collection of modest sets via the small set
of PERs.
Specifically,
every modest set is isomorphic to one induced by a PER.
This equivalence forms the basis of the impredicative universe in our realizability model.

\begin{example}
\label{exa:impred-universe}
Given an assembly $X \in \Asm{\LCAtoPCA{\LCA}}$,
we define the family $\Univ$ over $X$ to be the set of partial equivalence relations over $\LCA$.
If $R$ is a PER,
then $a \realize R$ always holds.
Let $t$ be a term of type $\Univ$ in context $X$.
Since every PER gives rise to a modest set,
$t$ induces a cartesian type $\ElC{t}$ in context $X$,
and we define the linear $\ElL{t}$ in the same way.

The fact that $\Univ$ is an impredicative universe
follows from two observations.
First,
we note that by \Cref{thm:modest-set-equiv} every modest set is isomorphic
to a modest set induced by a PER.
Second,
modest sets are closed under both dependent products and linear dependent products.
Specifically,
if we have assemblies $Y_x$ and $Z_y$ for each $x : X$
and $y : \total{Y}$,
then their dependent product is modest if each $Z_y$ is modest.
One can prove a similar statement for linear dependent products.
\end{example}

\section{Impredicative encodings}\label{sec:encodings}

Finally,
we show how one can use impredicative encodings in linear dependent type theory
to construct initial algebras.
Recall that in second-order dependent type theory with an impredicative universe $\Univ$,
one can encode inductive data types via their recursion principle~\cite{girard:1989}.
For example,
if we have some type $A$,
then we define the type of lists over $A$ to be
\[
\prod_{X : \Univ} X \to (A \rightarrow X \to X) \to X
\]
Since we assume $\Univ$ to be an impredicative universe,
this type lives in $\Univ$.
One can show that this type satisfies the recursion principle and computation rules for lists definitionally.

However,
such an impredicative encoding does not necessarily give rise to an initial algebra \cite{geuvers:2001,rummelhoff:2004},
which means that one cannot prove a suitable induction principle for this type.
There are various ways to solve this problem,
and one solution was given by Awodey, Frey, and Speight \cite{awodey:2018}.
Their idea is that a suitable refinement of this impredicative encoding satisfies the induction principle.
Specifically, the desired initial algebra can be constructed as a subtype of the given impredicative encoding.
This idea can be applied in a wide variety of settings.
For example, it works for higher inductive types, even without truncation restrictions on the eliminator \cite{ripollecheveste:2023}. Also, one can construct coinductive types \cite{bronsveld:2025}.

In the remainder of this section,
we adapt the ideas of Awodey, Frey, and Speight to construct initial algebras in linear dependent type theory.
Our focus lies on initial algebras in the category of linear types,
because the methods of Awodey, Frey, and Speight can be used directly
to construct initial algebras in cartesian types.
Throughout this section,
we focus on the type of lists whose type of elements is some closed linear type $A$.
Since we only work within the universe $\Univ$ in the remainder,
we use normal notation instead of $\Univ$-internal notation,
e.g. we write $\Sqcap$ instead of $\widehat{\Sqcap}$.

Following Girard \cite{girard:1989},
we start with the following encoding.
\[
\ListI{A} \defeq \LPi{X : \Univ}{\ElL{X} \lin \expL{(A \lin \ElL{X} \lin \ElL{X})} \lin \ElL{X}}
\]
We define the constructor $\NilI : \ListI{A}$ as follows.
\[
\NilI \defeq \lam{(X : \Univ) (n : \ElL{X}) (c : \expL{(A \lin \ElL{X} \lin \ElL{X})})}{n}
\]
If we have $a : A$ and $l : \ListI{A}$,
then $\ConsI{a}{l} : \ListI{A}$ is defined as follows.
\[
\ConsI{a}{l} \defeq \lam{(X : \Univ) (n : \ElL{X}) (c : \expL{(A \lin \ElL{X} \lin \ElL{X})})}{c \> a \> (l \> X \> n \> c)}
\]
To define the recursor,
we assume that we have a type $X : \Univ$, an element $n : \ElL{X}$, and an operation $c : \expL{(A \lin \ElL{X} \lin \ElL{X})}$,
and we define $\Rec{X}{n}{c}{l} : \ElL{X}$ as
\[
\Rec{X}{n}{c}{l} \defeq l \> X \> n \> c
\]

In general, the type $\ListI{A}$ only is a weakly initial algebra
and not necessarily initial~\cite{geuvers:2001,rummelhoff:2004}.
We obtain the desired initial algebra by taking a suitable subtype,
and for that reason,
we first define algebras and their morphisms.
\[
\Alg
\defeq
\sum_{X : \Univ} \m(\ElL{X}) \times \m(\expL{(A \lin \ElL{X} \lin \ElL{X})})
\]
An algebra is given by a type $X : \Univ$
together with cartesian terms of type $\m(\ElL{X})$ and $\m(\expL{(A \lin \ElL{X} \lin \ElL{X})})$.
Such cartesian terms correspond to linear terms of types $\ElL{X}$ and $\expL{(A \lin \ElL{X} \lin \ElL{X})}$ in the empty linear context.
We introduce the following notation:
if we have an algebra $X : \Alg$,
we write $n_X : \ElL{X}$
and $c_X : \expL{(A \lin \ElL{X} \lin \ElL{X})}$
for the structure maps.


Next we define morphisms between algebras.
Given  $f : \ElL{X} \lin \ElL{Y}$,
we define
\begin{align*}
&\IsMor{X}{Y}{f}
\defeq
(\lineq{f \> n_X}{n_Y})
\times
\prod (a : A) (x : X) . \lineq{f(c_X \> a \> x)}{c_Y \> x \> (f \> y)}
\\
&\Mor{X}{Y}
\defeq
\sum (f : \m(\ElL{X} \lin \ElL{Y})) . \IsMor{(X, n_X, c_X)}{(Y, n_Y, c_Y)}{f}
\end{align*}
Every cartesian term $f : \m(\ElL{X} \lin \ElL{Y})$ gives rise to a linear term of type $\mu(f) : \ElL{X} \lin \ElL{Y}$ in the empty context (c.f. the $\m$-\textsc{e} rule),
and, to lighten the notation,
we omit the $\mu$.
Thus,
every morphism $f$ of algebras gives rise to a linear term $f : \ElL{X} \lin \ElL{Y}$ in the empty context.
By injectivity and reflection,
the expected diagrams commute.


Now we are ready to define the refined encoding.
Awodey, Frey, and Speight construct the initial algebra using an equalizer,
which gives a subtype of  $\ListI{A}$.
Although they use $\sum$-types to construct the equalizer,
we refrain from the linear analogue,
namely $\LSigSym$-types.
The type $A$ in $\LSig{x : A}{B  \> x}$ is required to be cartesian,
so we cannot use $\LSigSym$-types to construct subtypes of linear types.
This is why we included equalizers in \Cref{sec:dltt}.

The necessary morphisms are defined as follows.
\begin{align*}
&\varphi , \psi : \ListI{A} \lin
\LPi{(X, n_X, c_X) : \Alg}{\LPi{(Y, n_Y, c_Y) : \Alg}{\LPi{(f, p, q) : \Mor{X}{Y}}{\ElL{Y}}}}
\\
&\varphi \defeq \lambda l \, X \, n_X \, c_X \, Y \, n_Y \, c_Y \, f \, p \, q.
f(\Rec{X}{n_X}{c_X}{l})
\\
&\psi \defeq \lambda l \, X \, n_X \, c_X \, Y \, n_Y \, c_Y \, f \, p \, q.
\Rec{Y}{n_Y}{c_Y}{l}
\end{align*}
Note that in the definition of $\varphi$,
the morphism $f$ refers to the corresponding linear term of type $\ElL{X} \lin \ElL{Y}$ in the empty linear context.

We define the type $\List{A}$ to be $\LEq{\varphi}{\psi}$.
Now we prove that $\List{A}$ is an initial algebra,
and we start by showing that that $\List{A}$ is an algebra.
To do so,
we prove that $\NilI$ and $::^*$ give rise to operations on $\List{A}$,
which is stated in the following lemma.

\begin{lemma}
\label{lem:list-alg}
We have $\varphi(\NilI) = \psi(\NilI)$.
In addition,
we have $\varphi(\ConsI{a}{l}) = \psi(\ConsI{a}{l})$
for all $a : X$ and $l : \List{A}$.
\end{lemma}

\begin{proof}
To show that $\varphi(\NilI) = \psi(\NilI)$,
we first note that
\begin{align*}
&\varphi(\NilI) = f(\Rec{X}{n_X}{c_X}{\NilI}) = f \> n_X
\\
&\psi(\NilI) = \Rec{Y}{n_Y}{c_Y}{\NilI} = n_Y
\end{align*}
Since $f$ is a morphism,
it is indeed the case that $f \> n_X = n_Y$ by reflection and injectivity.

For the other statement, we assume that
\[
\varphi \> l = f(\Rec{X}{n_X}{c_X}{l}) = \Rec{Y}{n_Y}{c_Y}{l} = \psi \> l.
\]
Now we have
\begin{align*}
&\varphi(\ConsI{a}{l}) = f(\Rec{X}{n_X}{c_X}{\ConsI{a}{l}}) = f(c_X \> a \> (\Rec{X}{n_X}{c_X}{l}))
\\
&\psi(\ConsI{a}{l}) = \Rec{Y}{n_Y}{c_Y}{\ConsI{a}{l}} = c_Y \> a \> (\Rec{Y}{n_Y}{c_Y}{l}) = c_Y \> a \> (f(\Rec{X}{n_X}{c_X}{l}))
\end{align*}
From injectivity,
it follows that $\lineq{\varphi(\ConsI{a}{l})}{\psi(\ConsI{a}{l})}$
since $f$ is a morphism.
\end{proof}

By \Cref{lem:list-alg},
we have an element $\Nil : \List{A}$
and a function sending $a : A$ and $l : \List{A}$
to $\Cons{a}{l} : \List{A}$.
From this
we get an algebra $\List{A}$
whose carrier is $\List{A}$
and whose algebra map is given by $\Nil$ and $::$.
We conclude this section by proving that $\List{A}$ is initial.

\begin{theorem}
\label{thm:initial-impred}
The algebra $\List{A}$ is initial.
\end{theorem}

\begin{proof}
Given an algebra $X$,
we define $f$ to be $\lambda (l : \List{A}) . \Rec{X}{n_X}{c_X}{l}$.
Note that
\begin{align*}
&\Rec{X}{n_X}{c_X}{\projl{\Nil}} = \Rec{X}{n_X}{c_X}{\NilI} = n_X
\\
&\Rec{X}{n_X}{c_X}{\projl{(\Cons{a}{l})}} = \Rec{X}{n_X}{c_X}{\ConsI{a}{l}} = c_X \> a \> (\Rec{X}{n_X}{c_X}{l})
\end{align*}

To prove that the map $f$ is unique,
it suffices to show that $\Rec{\List{A}}{\Nil}{\ConsSym}{l} = l$
for each $l : \List{A}$.
By function extensionality,
we need to show for all $X$, $n_X$, and $c_X$ that
\[
\Rec{\List{A}}{\Nil}{\ConsSym}{l} \> X \> n_X \> c_X = l \> X \> n_X \> c_X.
\]
By definition,
we have $l \> X \> n_X \> c_X = \Rec{X}{n_X}{c_X}{l}$
and that
\[
\Rec{\List{A}}{\Nil}{\ConsSym}{l} \> X \> n_X \> c_X
=
\Rec{X}{n_X}{c_X}{\Rec{\List{A}}{\Nil}{\ConsSym}{l}}.
\]
It suffices to show that
\[
\Rec{X}{n_X}{c_X}{\Rec{\List{A}}{\Nil}{\ConsSym}{l}} 
=
\Rec{X}{n_X}{c_X}{l}.
\]
Since $\lambda l. \Rec{X}{n_X}{c_X}{l}$ is a morphism,
the theorem follows from the fact that $l : \List{A}$.
\end{proof}








\section{Related work}\label{sec:related}


Linear combinatory algebras were first introduced in a series of lectures given by Abramsky \cite{Abramsky1997Interaction}; the first published account is \cite{AbramskyHaghverdiScott2002LCAs}, wherein there are many examples. The first linear realizability model is due to Abramsky and Lenisa \cite{AbramskyLenisa2005LinearRealizability}, who consider categories of partial equivalence relations over LCAs in order to obtain full completeness results for various typed $\lambda$-calculi. Hoshino \cite{hoshino:2007} offered the first in-depth study of assemblies and modest sets over (relational) linear combinatory algebras. Since then, there has been a growing interest in ``substructural'' combinatory algebras and their associated categorical structures, for example: ribbon combinatory algebras \cite{Hasegawa2021BraidedLambda,HasegawaLechenne2024Braids} and planar combinatory algebras \cite{Tomita2021RealizabilityWithoutSymmetry,Tomita2022PlanarRealizability,Tomita2023CategoricalRealizability}. Kuzmin, Nester, Reimaa and Speight \cite{KuzminNesterReimaaSpeight2025CombinatoryCompleteness} formulate a quite general theory of combinatory completeness using structured multicategories.
Recently, monadic combinatory algebras have been studied for effectful realizability \cite{cohen:2025a,cohen:2019,cohen:2025}.

Various treatments of linear dependent type theory are given throughout the literature.
The first development of LDTT was by Cervesato and Pfenning~\cite{cervesato1996linear,MR1946180},
but they only considered a subset of the desired type formers~\cite{MR1946180}.
Krishnaswami, Pradic, and Benton gave a version of LDTT
that contains all connectives from linear logic~\cite{DBLP:conf/popl/KrishnaswamiPB15}.
Similar treatments were given by V\'ak\'ar~\cite{vakar:2015}
and by Lundfall~\cite{lundfall:2017},
with some differences.
Compared to Krishnaswami-Pradic-Benton,
Lundfall gives different rules for identity types and the modality,
and V\'ak\'ar has a slightly different treatment of $\sum$-types.
We use a similar syntax as Krishnaswami, Pradic, and Benton.
The key differences are indicated in \Cref{sec:dltt}:
our syntax comes with equalizer types, an injectivity axiom for the modality $\m$, and an impredicative universe of types.

Both Lundfall~\cite{lundfall:2017} and  V\'ak\'ar~\cite{vakar:2015} give a notion of categorical model for LDTT.
Lundfall's treatment is based on comprehension categories,
whereas V\'ak\'ar uses indexed categories instead.
In addition,
Lundfall represents the modality of linear logic via an adjunction
whereas V\'ak\'ar uses a comonad.
We base our semantics on Lundfall's notion of linear comprehension category.
Neither Lundfall nor V\'ak\'ar consider a realizability model.
Krishnaswami, Pradic, and Benton~\cite{DBLP:conf/popl/KrishnaswamiPB15} give semantics to their calculus using partial equivalance relations,
but they do not provide a general development of realizability models.

Quantitative type theory~\cite{atkey:2018,mcbride:2016,nakov:2021}
is a generalization of linear dependent type theory
that refines the linearity condition by allowing variables to
appear in quantities indicated by some semiring.
Atkey~\cite{atkey:2018} constructed a model of quantitative type theory using a quantitative combinatory algebra.
The resulting model shows resemblance to ours,
a crucial difference being that our model is impredicative.
Since Atkey interprets types as assemblies indexed by a set rather than an assembly,
there cannot be an impredicative universe in that model~\cite{phoa:1992}.
Hence, the impredicative encoding in \Cref{sec:encodings} can only be interpreted in our model. 

The method of refining impredicative encodings that we use in \Cref{sec:encodings} is based on the work by Awodey, Frey, and Speight \cite{awodey:2018}, which treats (higher) inductive types.
The truncation restriction on eliminators is lifted in \cite{ripollecheveste:2023}.
The method has also been extended to coinductive types \cite{bronsveld:2025}.
Another way to fix impredicative encodings is via parametricity.
By assuming internal parametricity,
initiality for impredicative encodings can be proven~\cite{altenkirch:2024,wadler:2003,wadler:2007}.
Such axioms are validated by PER models of type theory~\cite{krishnaswami:2013}.
One can use impredicative encodings in higher-order separation logic to define inductive predicates~\cite{krebbers:2025}.
Since they use proof-irrelevant predicates,
they do not need to refine their encodings.

\section{Conclusion}\label{sec:conclusion}
In this paper,
we defined a realizability model of linear dependent type theory with an impredicative universe for cartesian and linear types, all parametrized by a linear combinatory algebra.
We also showed how the method of Awodey, Frey and Speight~\cite{awodey:2018} transfers to impredicative encodings of linear types, constructing the initial algebra of linear lists.

There are several ways to further develop this work.
One can further generalize the development in \Cref{sec:encodings}
by considering arbitrary inductive types rather than only lists.
We conjecture that the same ideas can be used.
In addition,
we use extensionality throughout this paper,
and our injectivity axiom is tailored for this setting.
One question is how to generalise our work to intensional versions of linear dependent type theory.

\bibliography{fscd-references}

\end{document}